%% file: ft.tex
\documentclass[11pt]{article}
\usepackage{etex}
\usepackage{makeidx}
\usepackage{tikz}
\usepackage[utf8]{inputenc}
\usepackage{authblk}
\usetikzlibrary{arrows,calc}
\usepackage{amsmath}
\usepackage{amsthm}
\usepackage{bbm}
\usepackage{url}
\usepackage{verbatim}
\usepackage{graphicx}
\usepackage{subcaption}
\usepackage{framed}

\usepackage[letterpaper,hmargin=1in,vmargin=1.25in]{geometry}
\newtheorem{definition}{Definition}
\newtheorem{theorem}{Theorem}
\newtheorem{lemma}{Lemma}
\newtheorem{corollary}{Corollary}
\theoremstyle{remark}
\newtheorem{rmk}{Remark}

\newcommand{\nc}{\newcommand}
\nc{\enc}[1]{\ket{\overline{#1}}}
\usetikzlibrary{matrix,backgrounds,fit,decorations.pathreplacing}
\nc{\ketbra}[2]{\vert#1\rangle\!\langle#2\vert}
\nc{\kb}[1]{\vert#1\rangle\!\langle#1\vert}
\nc{\braket}[2]{\langle#1\vert#2\rangle}
\nc{\tr}{\operatorname{tr}}
\nc{\cnot}{\operatorname{CNOT}}
\nc{\Init}{\operatorname{Init}}
\nc{\Eval}{\operatorname{Eval}}
\nc{\leftbrace}[2]{#1\left\{\vphantom{\begin{matrix} #2 \end{matrix}}\right.}
\nc{\rightbrace}[2]{\left.\vphantom{\begin{matrix} #2 \end{matrix}}\right\}#1}
\nc{\flr}{\ensuremath{\mathcal{F}_{\text{LR}}}}
\nc{\shor}{\ensuremath{\ket{\text{even}_7}}}

\usepackage[colorlinks=true, citecolor=blue, hyperindex] {hyperref}
\usepackage{todonotes}

\usepackage{etoolbox}

\makeatletter
\newcommand\appendix@section[1]{%
  \refstepcounter{section}%
  \orig@section*{Appendix \@Alph\c@section: #1}%
  \addcontentsline{toc}{section}{Appendix \@Alph\c@section: #1}%
}
\let\orig@section\section
\g@addto@macro\appendix{\let\section\appendix@section}
\makeatother

\tikzstyle{int} = [draw,minimum height=10em,minimum width=4em]
\tikzstyle{party} = [draw=none, minimum height=2em,minimum width=2em]
\tikzstyle{arr} = [->,minimum size=2em]

\input{Qcircuit}

\title{Classical leakage resilience from fault-tolerant quantum computation}
\author[1,2]{Felipe G. Lacerda\thanks{gomes@phys.ethz.ch}}
\author[1]{Joseph M. Renes\thanks{renes@phys.ethz.ch}}
\author[1]{Renato Renner\thanks{renner@phys.ethz.ch}}
\affil[1]{Institute for Theoretical Physics, ETH Zurich, Switzerland}
\affil[2]{Departamento de Ciência da Computação, Universidade de Brasília,
  Brazil}
\date{\vspace{-5ex}}

\begin{document}
\maketitle

\begin{abstract}
  Physical implementations of cryptographic algorithms leak information, which
  makes them vulnerable to so-called side-channel attacks. The problem of secure
  computation in the presence of leakage is generally known as leakage
  resilience. In this work, we establish a connection between leakage resilience
  and fault-tolerant quantum computation. We first prove that for a general
  leakage model, there exists a corresponding noise model in which fault
  tolerance implies leakage resilience. Then
  we show how to use constructions for fault-tolerant quantum computation to 
  implement classical circuits that are secure in specific leakage models. 
\end{abstract}

\section{Introduction}
\label{sec:intro}

Modern theoretical cryptography is primarily concerned with developing schemes
that are \textit{provably secure} under reasonable assumptions. While the field
has been hugely successful, the threat model considered usually doesn't allow
for the possibility of \textit{side-channel attacks}---attacks on the physical
implementation of the cryptographic scheme.

Side-channel attacks have been a worry long before the advent of modern
cryptography. As early as 1943, it was discovered that a teletype used for
encryption by the American military caused spikes in an oscilloscope that could
then be used to recover the plaintext~\cite{nsatempest}. More recently,
side-channel attacks on cryptographic applications widely used in practice have
been revealed. One of these is the ``Lucky Thirteen'' attack on TLS in CBC
mode~\cite{alfardan2013lucky}, which is based on measuring the time it takes the
server to reply to a request over the network. Another recent
attack~\cite{genkin2013rsa} uses acoustic cryptanalysis to attack the GnuPG
implementation of RSA. The authors managed to extract the full RSA key by
measuring the noise produced by the computer while it decrypts a set of chosen
ciphertexts. Although the relevant software has since then been updated so that
these attacks are no longer possible, they highlight the importance of designing
implementations with side-channel attacks in mind.

The theoretical approach to side-channel attacks is to design \textit{protocols}
that are resilient against them. This is the focus of the area known as
\textit{leakage resilience}. In this work, we present a way to perform universal
leakage-resilient computation, i.e.\ we construct a general
``leakage-resilient compiler'' that takes an arbitrary circuit and produces a new, leakage-resilient 
version having the same computational functionality.

We take a novel approach to leakage resilience: fault-tolerant
quantum computation.  The basic idea is that all actions performed in the
execution of a classical circuit---as well as leakage attacks on it---can be
described using the formalism of quantum mechanics. Then, leakage of the
physical state of the computation is equivalent to a so-called phase error in a
quantum circuit.  Since a fault-tolerant quantum computation must protect
against phase errors (as well as more conventional bit-flip errors), it is
necessarily leakage-resilient.  However, achieving leakage resilience in this
way would require a fault-tolerant quantum circuit.  Here, we give a method for
constructing a leakage-resilient \emph{classical} circuit by modifying an
appropriate fault-tolerant \emph{quantum} circuit. The former mimics the latter,
inheriting its leakage resilience.  Our approach is similar to security proofs
of quantum key distribution in which it is shown that the operation of the
actual protocol---where the outcomes of measurements on entangled quantum states
are processed classically---can be interpreted as mimicking a fully quantum
protocol for entanglement distillation, from which the protocol derives its
security~\cite{shor2000simple}.

Before introducing our general setting for leakage, it is illustrative to
present a concrete example of leakage resilience in practice. Smart cards are
integrated circuits that have been widely used for authentication and also allow
the storing of sensitive data, making them portable carriers of information such
as money and medical records. Since smart cards are designed to be portable,
they are subject to a variety of physical attacks. Possible attacks include
measuring the time and the electrical current used when performing operations
(power analysis). Thus, leakage resilience is essential to the design of smart
cards.

More specifically, a smart card stores some internal data, provides an interface
for external input, performs some computation on the internal data and the
external input and sends the output through an output interface. One of the
design goals is that if the smart card is given to an adversary, who can send
inputs to and read off outputs from it, the adversary should not be able to
obtain any information about the internal data in the card, beyond what can be
learned from the regular output of the computation. This design goal is what we
mean by leakage resilience. However, whereas smart cards are usually designed to
be resilient against specific attacks, we are interested in larger classes of
attacks. Additionally, smart cards are devices that perform one fixed function,
while our goal is to have leakage resilience for general functions.

Our general setting for leakage is as follows: first, an honest party (Alice)
inputs a circuit $\mathcal{C}$ that computes a function $f$ along with a secret
$y$. Then, the adversary (Eve) is given black-box access to the circuit, so that
she can interactively send inputs, denoted $x$, to the circuit and gets the
corresponding outputs ($f(x,y)$). She also receives a description of the
circuit. Additionally, with each interaction Eve obtains leaked bits according
to the leakage model. Roughly, the goal is that Eve should not learn any more
information about $y$ than what she would get from $f(x,y)$ alone. This setting
will be formalized in Section~\ref{sec:def} using the abstract cryptography
framework~\cite{maurer2011abstract,maurer2012constructive}.

\subsection{Previous work}

The question of hiding internal computation from an eavesdropper is related to
the problem of \textit{program obfuscation}. Obfuscation can be seen as a
``worst-case leakage resilience'', in which the internal state must be protected
even if the whole execution leaks to the adversary.\footnote{One explicit
  connection between obfuscation and a specific leakage model (the ``only
  computation leaks information'' model of Micali and
  Reyzin~\cite{micali2004physically}) has been made
  in~\cite{bitansky2011program}; see also~\cite{goldwasser2012compute}.}
However, it is known that obfuscating programs is impossible in
general~\cite{barak2001possibility}. Thus, if we want any leakage resilience at
all, the type of leakage allowed has to be restricted in some way.

This is not necessarily a problem because in practice the adversary has its own
limitations when trying to obtain information from the system. Here we list a
few results in the field, although the list is by no means
comprehensive. Ishai~\textit{et al.}~\cite{ishai2003private} considered
adversaries that can learn the values of a bounded number of wires in the
circuit. Micali and Reyzin~\cite{micali2004physically} introduced the ``only
computation leaks information'' assumption, in which the leakage at each step of
the computation only depends on data that was used in the
computation. Faust~\textit{et al.}~\cite{faust2010protecting} consider one model
where the adversary gets a function of all of the circuit's state (that is, the
output and all intermediate computations), the restriction being that the
leakage function must be computable in $\textrm{AC}^0$, the family of circuits
containing only AND, NOT and OR gates (with unbounded fan-in) and having
constant depth and polynomial size. They also consider noisy leakage, where the
whole state of the circuit leaks, but the adversary only receives a noisy
version of the leaked state, each bit being flipped with some probability. A
common further assumption is to use a small component of
trusted
hardware~\cite{goldwasser2010securing,faust2010protecting,dziembowski2012leakage}.

Despite these advances, relating leakage models to actual leakage seen in
practice has so far proved to be challenging. Standaert \textit{et
  al.}~\cite{standaertleakage} present a few problems with the ``bounded
leakage'' assumption, in which the leakage at each computation step is assumed
to be bounded; see also~\cite{standaert2010leakage}.

Our work is loosely related to the one of Smith \textit{et~al.}~\cite{crepeau2002secure}, where techniques from fault-tolerant quantum
computation are used to develop a construction for secure multi-party quantum
computation. Given that multi-party computation techniques are commonly used for
leakage resilience (in particular, secret
sharing~\cite{ishai2003private,dziembowski2012leakage,goldwasser2010securing}),
the connection between fault tolerance and leakage resilience is perhaps unsurprising. 
However, the construction of Smith~\textit{et~al.\@} was inherently
quantum, whereas our construction runs on a classical machine.

\subsection{Our contribution}
We establish a relation between leakage-resilient (classical) computation and
fault-tolerant quantum computation, which are formally defined in
Section~\ref{sec:def}.  Specifically, we show how methods of the latter
can be used to construct leakage-resilient compilers, which transform a given
circuit into another (classical) circuit with the same computational
functionality as well as resistance to leakage.

The starting point in relating leakage-resilient classical computation and
fault-tolerant quantum computation is the observation that any classical logical
operation can be regarded as a quantum operation performing the same action in
the so-called computational basis. (Appendix~\ref{sec:app} provides a brief
background on the formalism and tools of quantum information theory necessary
here.) Then, as we show in Section~\ref{sec:phase-noise}, any given leakage
model may be interpreted as specific model of \emph{phase noise} afflicting the
corresponding quantum circuit.

Phase noise is not the most general type of quantum noise. Nonetheless, as we
show in Theorem~\ref{thm:ft-lr}, if fault tolerance is possible for a given
noise model---meaning roughly that error correction is performed frequently
enough that the encoded information essentially never suffers from errors---the
quantum computation is resilient to leakage of the corresponding leakage model.
However, we want to make classical circuits leakage-resilient and do not
necessarily want to carry out a quantum computation to achieve this goal.
Fortunately, the structure of certain quantum error-correcting codes is such
that we can mimic the error-correcting steps with classical circuits.

Following this approach, we construct a general leakage-resilient compiler by
mimicking the basic fault tolerant components of the fault-tolerant scheme
in~\cite{aliferis2005quantum} using classical components.  Section~\ref{sec:red}
describes how to transform into classical circuits certain simple types of
quantum circuits that serve as building blocks for arbitrary
circuits. Section~\ref{sec:lrg} then presents a fault-tolerant implementation of
the Toffoli gate, a gate that is universal for classical computation. Combining
this with the results of Section~\ref{sec:red} then gives a leakage-resilient
compiler.

As in other works, our construction assumes the existence of a part of the
circuit that is leak-free. The assumption is however minimal: the only leak-free
component we use is a source of (uniformly) random bits. Using the construction,
one can transform an arbitrary classical circuit into a circuit that is
resilient to leakage arising in any model for which reliable quantum computation
is possible under the corresponding phase noise model. One particular leakage
model that translates into a well-studied quantum noise model is that of
independent leakage, in which the value of each wire of the circuit leaks with
some fixed probability. While potentially too restrictive (in particular, the
independence assumption implies that the adversary does not choose the wires
that leak), its simplicity allows for an easy interpretation in the quantum
scenario as independent phase errors, for which various fault-tolerant
constructions are known. This model is used in our construction in
Section~\ref{sec:lrg}. Leakage models that include correlations lead to
error models that have yet to be analyzed.

We stress that rather than presenting a specific leakage resilient scheme, the
contribution of this work aims to provide a novel approach to
leakage resilience and connect this field of cryptography to the research area
of fault-tolerant quantum computation. Our work thus shows how results achieved
in one area (e.g., new threshold theorems for quantum fault tolerance) can be
translated to the other (e.g., bounds on the performance of leakage resilient
compilers). Our hope is that our result will inspire research developing the
relationship between fault tolerance and leakage resilience. In
Section~\ref{sec:thr} we discuss possible future directions for this line of
work.

\section{Definitions}
\label{sec:def}

\subsection{Abstract cryptography framework}
\label{sec:abstr-crypt}

In order to define leakage resilience we use the \textit{abstract cryptography}
framework~\cite{maurer2011abstract,maurer2012constructive}, of which we give a
short summary. From an abstract viewpoint, constructing a protocol amounts to
assuming that a certain set of \textit{real resources} is available and then
using them to build a new resource, termed an \textit{ideal resource}. By way of
composition, the ideal resource can then again be used as a real resource in
another protocol to build a more complex ideal resource.

For instance, the one-time pad construction assumes that a resource giving out a
secret key is available along with an authentic channel. It is then shown that
these resources can emulate a secure channel. In this case, the ideal resource
is the secure channel, whereas the authentic channel along with a shared secret
key is used as a real resource. On the other hand, in a protocol for
authentication, the authentic channel takes the role of the ideal resource, and
the real resource is a completely insecure channel together with a secret key.

In this framework, resources are a type of \textit{systems}, which are defined
as abstract objects which can be composed. Each system has an interface set
$\mathcal{I}$, and interfaces can be connected in order to form new
systems. Resources are systems where each interface corresponds to one party
that has access to it. As an example of a resource, we can define a private
channel between honest parties $A$ and $B$ subject to possible eavesdropping by
$E$ as a resource that takes inputs from $A$ and outputs them at $B$. Since the
resource gives $E$ no outputs, the channel is private by definition. Note that
this holds whether $E$ acts honestly or dishonestly; in abstract cryptography
the goal is to emulate the behavior of ideal resources in all situations, not
just when assuming certain parties are honest and others dishonest.

To emulate an ideal resource from given resources, the latter may be
composed. Furthermore, each party can act on their interface using a
\emph{converter}, which is also modelled as a system. A converter has an
``inside'' interface, connected to the resource, and an ``outside'' interface,
which is used by the parties. A \emph{protocol} specifies a converter for each
 party acting honestly and is applied to a real resource. In the security arguments, we
will also consider converters applied to the ideal resource by parties acting dishonestly. 
These are termed \emph{simulators}. We denote composition of systems by
juxtaposition, so that, for instance, the resource formed by plugging converter
$\pi_A$ into the resource $\mathcal{R}$ is denoted by $\pi_A \mathcal{R}$.

In order to allow for constructions that do not perfectly match a desired ideal
resource but only approximate it, we need a notion of distance $d$ between two
resources, which must be a pseudo-metric $d$ on the set of resources. Typically
we consider the distance to be the maximal advantage that a system trying to
distinguish between the two resources (the \textit{distinguisher}) can
have. Given two resources $\mathcal{R}$ and $\mathcal{S}$ and a distance $d$, we
also write $\mathcal{R} \approx_\varepsilon \mathcal{S}$ to denote
$d(\mathcal{R},\mathcal{S}) \leq \varepsilon$.

In this work we only ever need to consider two parties, $A$ and $E$, where $A$
is assumed to be honest. To keep the formal treatment simple, we
restrict ourselves to this case. In this scenario, what it means for a protocol
to securely emulate an ideal resource $\mathcal{S}$ given a resource
$\mathcal{R}$ takes a particularly simple form. It reduces to two different
conditions, one where $E$ is dishonest and another where $E$ is honest---that is,
it follows the protocol. We use the following definition of security, which, for
our scenario, is sufficient to imply the definition presented
in~\cite{maurer2011abstract}.\footnote{We note, however, that in the general case
  one needs to use additional constructions, termed \textit{filters}. We refer
  to~\cite{maurer2012constructive,maurer2011abstract} for more details.}

\begin{definition}
  \label{defn:sec}
  Let $\mathcal{R}$ and $\mathcal{S}$ be resources with interface set
  $\mathcal{I} = \{A,E\}$. We say that a protocol $\pi = \{\pi_A,\pi_E\}$
  \textnormal{securely constructs $\mathcal{S}$ from $\mathcal{R}$ within
    $\varepsilon$} if there exists a simulator $\sigma_E$ such that
  \begin{align}
    \label{eq:sec1}
    \pi_A \mathcal{R} &\approx_\varepsilon \sigma_E \mathcal{S} \quad\text{and}\\
    \label{eq:sec2}
    \pi_A \pi_E \mathcal{R} &\approx_\varepsilon \mathcal{S}.
  \end{align}
\end{definition}

\subsection{Leakage resilience}
\label{sec:leakage-resilience}

Using the abstract cryptography framework, we can now define leakage resilience
by specifying the corresponding ideal and real resource. They are two-party
resources, with parties that we denote by Alice ($A$) and Eve ($E$), where Alice
is assumed to be honest.

\noindent \textbf{Ideal resource.}  Our goal is to be able to compute a function
that takes an input to be given by the adversary, and an additional input that
corresponds to the initial secret. First we describe intuitively the kind of
resource that we want. At the beginning, Alice inputs the circuit to be executed
along with the secret input. Then, a description of the circuit is given to Eve,
who can execute it as a black box freely.

In light of this informal description, we define the ideal resource
$\mathcal{S}$ as follows. Alice initially inputs a secret $y$ and a description
of a circuit $\mathcal{C}$ that evaluates a function $f$. Eve can send inputs
$x$, to which she receives outputs $f(x,y)$. The ideal functionality also
outputs $\mathcal{C}$ to Eve. The resource is shown in
Figure~\ref{fig:lr-ideal}.

What this definition implies is that information about $y$ can only leak through
$f(x,y)$ and $\mathcal{C}$. As a concrete example, let $\mathcal{C}$ be a
circuit implementing an encryption algorithm $f$ that encrypts inputs $x$ using
the secret key $y$. In establishing security for a cryptographic scheme, one
assumes that the secret key is completely hidden from Eve. In the
example, this assumption is ensured by the way the ideal resource is defined,
since the only way information about $y$ could leak to Eve is through
the ciphertext $f(x,y)$.

This example also illustrates that security and leakage resilience are separate
goals: we place no restrictions on $f$, so even if the function reveals some
information about $y$ that would be of no consequence for leakage
resilience. Instead, leakage resilience only ensures that no additional
information about $y$ leaks.

\vspace{.5em}

\noindent \textbf{Real resource.} Informally, the real resource (which we denote
by $\mathcal{R}$) is a ``leaky'' version of the ideal resource, in which
additional information becomes available to Eve. Alice's interface allows her to
input a secret $y$ and a circuit $\mathcal{C}$ that evaluates a function
$f$. Eve's interface allows her to send inputs $x$, to which she gets outputs
$\mathcal{C}(x,y)$, and additionally send \textit{leakage requests} $l$, getting
leakages $l'$. As in the ideal resource, Eve also gets $\mathcal{C}$. The idea
is that Eve can use $\mathcal{R}$ as a black box, but also obtain additional
information from the leakage which might reveal something about the secret. This
scenario is represented in Figure~\ref{fig:lr-real}.

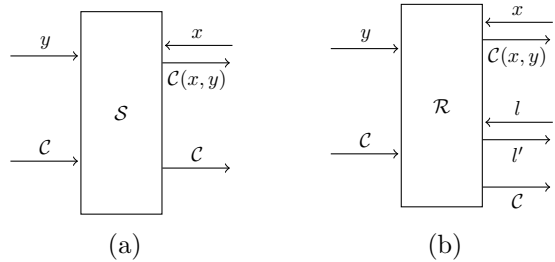
\begin{figure}[h!]
  \centering
  \begin{subfigure}[b]{0.25\textwidth}
    \input{lr-ideal.tex}
    \caption{}
    \label{fig:lr-ideal}
    \end{subfigure}
  \begin{subfigure}[b]{0.25\textwidth}
    \input{lr-real.tex}
    \caption{}
    \label{fig:lr-real}
\end{subfigure}
\caption{(\subref{fig:lr-ideal}) ideal resource for leakage resilience and
  (\subref{fig:lr-real}) real resource. In both cases, Alice has access to the
  left interface and Eve has access to the right interface.}
\end{figure}

In order to satisfy the conditions in Definition~\ref{defn:sec}, we need a
protocol $\pi$ to construct the ideal resource $\mathcal{S}$ from the real
resource $\mathcal{R}$. The first condition in Definition~\ref{defn:sec} is
depicted in Figure~\ref{fig:lr-sim}.

\begin{figure}[h!]
  \centering
  \scalebox{.9}{\input{lr-real-protocol} {\large $\approx_{\varepsilon}$}\input{lr-sim}}

  \caption{The first condition of Definition~\ref{defn:sec} applied to leakage
    resilience. In order to prove that the converter $\pi_A$ is part of a
    leakage-resilient protocol $\pi$, one has to show the existence of a
    simulator $\sigma_E$ such that execution of $\pi_A$ running with the real
    resource $\mathcal{R}$ is indistinguishable from execution of the ideal
    resource $\mathcal{S}$ with $\sigma_E$.}
  \label{fig:lr-sim}
\end{figure}
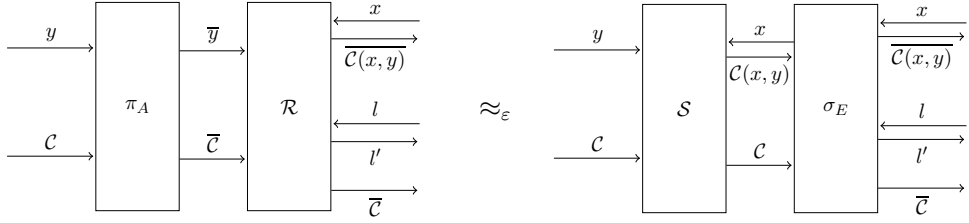

Our description of the real resource captures leakage in a very general form, in
the sense that $\mathcal{C},l$ and $l'$ may be arbitrary. However, in order for
leakage resilience to be possible, it is necessary to specify how the leakage
requests $l$ are chosen. Hence, we construct schemes that are resilient against
leakage with some particular structure. In order to formalize this, we'll make
the scenario presented above more concrete, by restricting the interactions in
the following way. Eve interacts with the resource in rounds. In each round, she
inputs pairs $(x,l)$ and receives outputs $(f(x,y),l')$.\footnote{We could have
  the function output be additionally a function of $l$; this situation would
  capture \textit{tampering}, in which Eve can introduce faults to the
  wires in the circuit~\cite{ishai2006private}.}

The particular strategy used to choose the leakage requests $l$ is referred to
as a \textit{leakage model} $L$. The leakage model is a set of allowed leakage
requests. A leakage request is a probability distribution over leakage
functions. That is, if the set of possible leakage functions is $\{l_i\}$, Eve
sends $l = \{(l_i,p_i)\}$, where $l_i$ is chosen with probability $p_i$. Each
$l_i$ is a function of the set of values assigned to the wires in the circuit,
which we denote by $W_{\mathcal{C}}(x,y)$ (when the circuit is given inputs $x$
and $y$). If the function $l_j$ is chosen, Eve receives $l' =
l_j(W_{\mathcal{C}}(x,y))$. Not that the output of $l_j$ may or may not
include the index $j$, depending on the leakage model.

As an example, Eve could choose at every round a bounded number of
wires of the circuit to leak. (These are the ``probing attacks'' considered
in~\cite{ishai2003private}.) Or one could restrict Eve to leakage
functions that are computed with constant-depth
circuits~\cite{faust2010protecting}. In any case, we define leakage-resilient
compilers with respect to a particular leakage model in the following way. In
the definition below, $\mathcal{S}$ and $\mathcal{R}$ are respectively the ideal
resource and the real resource defined above.

\begin{definition}
  The protocol $\pi$ is an \textnormal{$\varepsilon$-leakage-resilient compiler
    against leakage model $L$} if it securely constructs $\mathcal{S}$ from
  $\mathcal{R}$ within $\varepsilon$, where the leakage requests are drawn from
  $L$.
\end{definition}

\begin{rmk}
  \label{rmk:pie}
  When constructing the protocol $\pi = \{\pi_A, \pi_E\}$, we'll focus our
  attention on the converter $\pi_A$. This is due to the fact that given an
  appropriate $\pi_A$, it's always easy to construct a converter $\pi_E$ such
  that Condition~\eqref{eq:sec2} in Definition~\ref{defn:sec} is satisfied. We
  construct $\pi_E$ as follows. $\pi_E$ relays inputs $x$ from Eve to
  $\mathcal{R}$ and decodes the encoded output $\overline{\mathcal{C}(x,y)}$,
  which is then passed to Eve. It also decodes the encoded circuit $\mathcal{C}$
  before passing it to Eve. Crucially, $\pi_E$ does not provide the leakage
  interface (receiving $l$ and outputting $l'$) to Eve, effectively shielding
  the leakage from her.
\end{rmk}

\vspace{.5em}

\noindent \textbf{Independent leakage.} In this paper (Section~\ref{sec:lrg}) we
provide as a concrete example an explicit leakage-resilient compiler for a
concrete leakage model which we call \textit{independent leakage}. Independent
leakage is characterized by having every wire in the circuit potentially leak,
each with a fixed probability $p$.

This model can be formalized in the following way. Let $n$ be the number of
wires in the circuit. We label each wire in the circuit with an index $i$ with
$1 \leq i \leq n$. Now let $w$ be a binary string of length $n$ and $l_w(x,y)$
be the values of the wires $i$ with $w_i = 1$ when the circuit has $(x,y)$ as
input, along with the string $w$ to indicate which wires have leaked. For every
round of interaction, the probability $p_w$ that the leakage function $L_w$ was
chosen is given by $p_w = \Pr(X = |w|)$ where $X$ follows a binomial
distribution with parameters $n$ and $p$. The leakage model $L$ is then $L =
\{(l_w,p_w) \colon w \in \{0,1\}^n\}$. (Note that there's only one possible
leakage request.)

\vspace{.5em}

\noindent \textbf{Leak-free components.}  Many of the leakage-resilient
constructions require a small component of the circuit to be trusted or
``leak-free'', meaning its internal wires don't leak to Eve. For our
construction, this includes wires coming out of leak-free components. We
incorporate this requirement in our definition of leakage resilience by adding
the restriction that the leakage function $l$ may not depend on any wires inside
the leak-free components, or the wires coming out of them.

The only leak-free component we use in this work is a source of random bits: a
gate that takes no input and outputs a uniformly distributed random bit (which
is assumed not to leak). This is used in the part of our proof where we
transform a quantum circuit into a classical one; we leave open the question of
whether this is necessary in general. We note that our requirement is different
from than the one used in~\cite{faust2010protecting}: we require only uniformly
distributed bits, while their transformation requires bits distributed according
to an arbitrary (although fixed) distribution. On the other hand, the output
wires from their leak-free component is allowed to leak, whereas we require that
leak-free components not leak before interacting with other components.

\subsection{Fault tolerance}

A noisy quantum circuit $\mathcal{E}$ is an implementation of a unitary $U$
acting on subsystems $S$ (the \textit{data} subsystem) and $E$ (the
\textit{environment}). The idea of fault tolerance is to perform computations on
a noisy circuit reliably, by using the noisy circuit to execute encoded
operations such that at any point of the computation, the encoded state of the
data is close to the state of an ideal circuit that implements $U$ perfectly.

More precisely, let $\mathcal{C}$ be a quantum circuit with inputs and outputs
in $\mathcal{H}_2^{\otimes k}$. We assume that the circuit acts on classical
inputs; this is done by preparing a known state $\ket{0}^{\otimes k}$ and then
encoding a classical input $x$ into it, obtaining the state denoted
$\ket{x}$. The possible input states $\ket{x}$ form the computational
basis. Then, for any $x$, the circuit implements the action of a unitary $U_x$
on $\ket{0}^{\otimes k}$.

A fault-tolerant simulation for $\mathcal{C}$ works as follows. Let
$\tilde{\mathcal{C}}$ be a noisy quantum circuit acting on
$\mathcal{H}_2^{\otimes n}$. The noise incurred on the circuit is specified by a
\textit{noise model} $\mathcal{N}$, the form of which is defined in
Section~\ref{sec:phase-noise}. The data subsystem is initialized in the state
$\ket{0}^{\otimes n}$. The input $x \in \{0,1\}^k$ is then encoded in a quantum
error-correcting code of length $n$; we denote the encoded input by
$\enc{x}$. Execution then proceeds the same way as in the ideal circuit, except
that the gates are replaced by encoded gates (i.e., operations on encoded
data). We can thus compare the states of $\mathcal{C}$ and $\tilde{\mathcal{C}}$
at an arbitrary step of the computation. Additionally, after each step of the
computation, error correction is performed in order to keep the state of the
circuit in the encoded space of the code.  We say that $\tilde{\mathcal{C}}$ is
\textit{$\varepsilon$-reliable against noise $\mathcal{N}$} if for every step,
the state of $\mathcal{C}$ is equal to the logical value of the state of
$\tilde{\mathcal{C}}$ except with probability $\varepsilon$.

We note that the way we use fault tolerance is slightly different from the usual
treatment. As defined above (and as it's commonly done), the inputs are encoded
as part of the circuit. But our circuits also receive additional encoded
inputs. This is done purely for convenience and does not make the definition
stronger.

The focus in fault tolerance is in implementing so-called
\textit{gadgets}---components such as logical gates and error correction, that
can then be used as building blocks to construct reliable circuits. The goal is
usually to implement a set of gates that is universal for quantum
computation. However, as we will see in Section~\ref{sec:lrg}, since we only
seek to perform classical computation, we only need a restricted set of gates.

\section{Leakage models and quantum noise}
\label{sec:phase-noise}

In this section we show that for an arbitrary classical circuit with leakage
according to some leakage model $L$, we can view the circuit as a noisy quantum
circuit with a corresponding noise model $\mathcal{N}$, and that if the quantum
circuit is reliable with respect to $\mathcal{N}$ then it is also
leakage-resilient against $L$. To this end, we must first make matters more
concrete. A general quantum noise model is an operation $\mathcal N$ on quantum
states on $\mathcal{H}_2^{\otimes n}$ and takes the form
\begin{align}
	\label{eq:phasecptp}
  \mathcal{N}(\rho) = \sum_k p_k E_k \rho E_k^{\dagger},
\end{align}
where $E_k$ are arbitrary operators taking $\mathcal{H}_2^{\otimes n}$ to itself and $p_k\geq 0$ with $\sum_k p_k=1$. If
$E_k$ has the form $E_k = \bigotimes_{i=1}^n Z_i^{a_{i,k}}$ for $a_{i,k} \in
\{0,1\}$ and $Z$ specified by $Z\ket{x}=(-1)^x\ket{x}$, 
then we call $\mathcal{N}$ a \textit{phase noise} model.

The following lemma relates classical leakage models and quantum noise
models.

\begin{lemma}
	\label{lem:phasenoise}
    For any circuit $\mathcal{C}$ running with a leakage request
    $\{l_j,p_j\}_{j=1}^m$ taken from a leakage model $L$, there exists a noise
    model $\mathcal{N}$ such that $\mathcal{C}$, viewed as a quantum circuit, is
    a quantum circuit subject to noise $\mathcal{N}$. Furthermore, $\mathcal{N}$
    has the form
    \begin{equation}
      \label{eq:noise}
      \mathcal N(\rho)=\frac 1d \sum_{j=1}^m p_j \sum_{k=0}^{d-1}F_j^k \rho
      (F_j^k)^\dagger,
    \end{equation}
    where $F_j^k$ is the operator taking $\ket{s}$ to $\omega^{k l_j(s)} \ket{s}$
    and $\omega$ is a primitive $d$th root of unity.
\end{lemma}

Now we can state the formal connection between leakage resilience and reliable
quantum computation:

\begin{theorem}
  \label{thm:ft-lr}
  For every circuit $\mathcal{C}$ and leakage model $L$, there exists a noise
  model $\mathcal{N}$ as specified in Lemma~\ref{lem:phasenoise} such that for
  any $\varepsilon$-reliable (against $\mathcal{N}$) implementation
  $\tilde{\mathcal{C}}$ of $\mathcal{C}$ with encoding function $y \mapsto
  \enc{y}$, the protocol $\pi = \{\pi_A, \pi_E\}$ that, given $(y,\mathcal{C})$
  as input, outputs $(\enc{y},\tilde{\mathcal{C}})$, and $\pi_E$ is as specified
  in Remark~\ref{rmk:pie}, is a $2\sqrt{\varepsilon}$-leakage-resilient compiler
  against $L$.
\end{theorem}

\begin{proof}[Proof of Lemma~\ref{lem:phasenoise}]
  First consider the case where $L$ is such that at each round, a single leakage
  function $l$ is chosen. Let $S$ be the subsystem representing the wire
  assignments $s := W_{\mathcal{C}}(x,y)$ in the circuit and $E$ be Eve's
  subsystem.  In quantum-mechanical terms, the action of the leakage is the
  transformation
 \begin{align}
 	\ket{s}^S\otimes \ket{0}^E\quad\to\quad \ket{s}^S\otimes\ket{l(s)}^E
 \end{align} 
 for each $s$. To determine the action of the leakage on system $S$ itself, consider an arbitrary 
 superposition state $\ket{\psi}^S=\sum_s \sqrt{p_s}\ket{s}^S$ for some probability distribution
 $p_s$. After applying the transformation and tracing out subsystem $E$, $\rho=\kb{\psi}$
 becomes
 \begin{align}
 	\rho\quad\to\quad\mathcal{N}(\rho)=\sum_{s}\sum_{s':l(s')=l(s)}\sqrt{p_sp_{s'}}\ketbra{s}{s'},
 \end{align}
 as all coherence is lost between parts of the state with different values of $l$. 
 But this is also the output state if the transformation were instead
 \begin{align}
 	\ket{s}^S\otimes \ket{0}^E\quad\to\quad \frac 1{\sqrt d}\sum_k \omega^{k l(s)}\ket{s}^S\otimes\ket{k}^E=\frac 1{\sqrt d}\sum_k F^k\ket{s}^S\otimes \ket{k}^E,
 \end{align}
 where $d$ is the size of the output of $l$, $\omega$ a primitve $d$th root of
 unity, and $F^k$ the operator taking $\ket{s}$ to $\omega^{k
   l(s)}\ket{s}$. The equivalence of these two maps can be seen by tracing out $E$. Thus, the noise model can
 be expressed as 
 \begin{align}
 \label{eq:noisemodel}
 \mathcal N(\rho)=\frac 1d\sum_{k=0}^{d-1}F^k \rho
 (F^k)^\dagger.\end{align}

 Now consider the case where the leakage function is chosen probabilistically
 from a set $\{l_j\}_{j=1}^m$, where $l_j$ is chosen with probability $q_j$. In
 this case, a corresponding noise model is a convex combination of noise
 operators of the form in~\eqref{eq:noisemodel}. That is, we can write
\begin{equation}
  \label{eq:noise-convex}
  \mathcal{N}(\rho) = \frac{1}{d} \sum_{j=1}^m q_j \sum_{k=0}^{d-1} F_j^k \rho (F_j^k)^\dagger
\end{equation}
where $F_j^k$ is the operator taking $\ket{s}$ to $\omega^{k l_j(s)}\ket{s}$.
\end{proof}

\begin{proof}[Proof of Theorem~\ref{thm:ft-lr}]
  Every leakage request from $L$ has a corresponding noise model as given by
  Lemma~\ref{lem:phasenoise}. The fault-tolerant implementation of the circuit
  $\mathcal{C}$ includes a compiled circuit $\overline{\mathcal{C}}$ as well as
  a method to encode inputs. The converter $\pi_A$ receives the secret $y$ as
  well as a circuit as input, and outputs the encoded secret $\enc{y}$ and
  $\overline{\mathcal{C}}$. Following the fault-tolerant construction, inputs
  $x$ received from $E$ are encoded in $\enc{x}$, the quantum
  circuit is used to compute $\enc{\mathcal{C}(x,y)}$ and the decoded output is
  sent back to $E$. We therefore have a scenario as in Figure~\ref{fig:lr-real}.

  We now have to show the existence of a simulator $\sigma$ such that
  condition~\eqref{eq:sec1} (Fig.~\ref{fig:lr-sim}) is satisfied. We'll prove
  that such a simulator exists by showing that for any step of the computation,
  the leakage received by $E$ is independent of the circuit's current state
  (that is, the intermediate values encoded by its wires).

  By Lemma~\ref{lem:phasenoise}, the leakage model results in noise
  $\mathcal{N}$ on the quantum circuit $\mathcal{C}$. Since $\mathcal{C}$ is
  $\varepsilon$-reliable against $\mathcal{N}$, there can be no 
  errors on encoded quantum information at any point of the computation except with probability
  $\varepsilon$ (otherwise, the state of the reliable circuit
  would be different from the state in the ideal circuit). Now observe that the action of the ideal circuit is a unitary operation,
 and by undoing this at the output of the actual circuit, the result would be an approximate identity channel (having trivial action) connecting
  the input of the circuit to its output. According to \cite[Theorem 3]{kretschmann_information-disturbance_2008}, the channel to the eavesdropper must therefore have an approximately constant output, regardless of the input to the circuit. (It would also be possible to use an uncertainty principle recently derived by one of us for this argument~\cite{renes_operationally-motivated_2014}.) 
  Specifically, since the circuit is $\varepsilon$-reliable, the circuit followed by the reverse of the ideal unitary action is $\varepsilon$ close to an identity channel in the relevant norm (see \cite{kretschmann_information-disturbance_2008} for details), and this implies that the eavesdropper's output is $2\sqrt{\varepsilon}$ close to a channel with a constant output. 
  Therefore, the state at any point in the circuit, which we denote by $\ket{s}$, is
  independent of $l(W_{\mathcal{C}}(x,y))$ (except with probability
  $2\sqrt{\varepsilon}$).

We now note that for the case of an adaptive
adversary, $l'$ can depend on previous inputs and outputs, as well as previous
leakages, since the choice of $l$ can depend on those values. But because those
values are also available to the simulator, and because $l'$ is independent of
the logical value of $s$, $l'$ can be generated by the simulator.

\end{proof}

\section{From quantum to classical circuits}
\label{sec:red}

Our goal is to have \textit{classical} leakage resilience, and fault-tolerant
techniques yield a circuit that in general does not have a classical
translation. In this section, we show how to make a particular set of quantum
components classical. Then, in Section~\ref{sec:lrg}, we use these components to
make a fault-tolerant implementation of a quantum gate that is universal for
classical computation---namely, the Toffoli gate. Due to the fact that the
components have a classical translation, we also have a leakage-resilient
classical gate.

In order to show the equivalence between the quantum and classical circuits,
we analyze each component in the following scenario: we assume that the
quantum component has $Z$ basis inputs (i.e., the inputs are classical), and
that after execution the outputs are measured in the $Z$ basis. We then show
that for each component there exists a classical circuit that, when given the
same inputs, gives the same outputs as the quantum component after measuring. We
then use the components to construct encoded \textit{gadgets} that will be used
to implement the Toffoli gate. Since the components have a classical
translation, the gadgets and the Toffoli gate have one as well.

In order to be able to combine the classical translation with
Theorem~\ref{thm:ft-lr}, we need it to preserve the form of the original quantum
circuit, so that the connection between leakage and noise model can be made
explicitly. This is achieved by having, for each component, a classical
translation that has the same wires as the quantum component.

For the classical scenario, we assume that we can generate random bits in a
leak-free manner. We'll see that this is needed in order to make state
preparation of $\ket{+}$ classical. We prove the following result.

\begin{theorem}
  \label{thm:c-q}
  Let $\mathcal{C}$ be a quantum circuit accepting classical states as input and
  containing only $X,Z$, CNOT, CZ and Toffoli gates, state preparation of
  $\ket{0}$ and $\ket{+}$ and measurement in the $X$ and $Z$ bases. Then each
  gate $G$ in the circuit can be replaced by a classical circuit $G'$ with the
  same wires as $G$. Furthermore, assuming a leak-free source of random bits,
  there exist procedures for state preparation of $\ket{0}$ and $\ket{+}$ and
  measurement in the $X$ and $Z$ bases such that each of them can be replaced by
  a classical circuit with the same wires. If all these replacements are made,
  then the resulting classical circuit $\mathcal{C'}$ gives the same outputs as
  applying $\mathcal{C}$ followed by measuring its outputs in the $Z$ basis.
\end{theorem}

\begin{proof}
  Out of the gates we use, $X$, CNOT and Toffoli are classical, and therefore
  their classical translations are trivial. The $Z$ gate flips the input's
  phase; as we're assuming the inputs and outputs to be classical, it
  corresponds classically to the identity gate. The same holds for CZ.

  There are two bare qubits we need to prepare: $\ket{0}$ and $\ket{+}$.  The
  classical translation of preparing a $\ket 0$ is just preparation of 0.  To
  determine the translation of preparing $\ket +$, note that immediate
  measurement in the $Z$ basis would yield a random bit.  We can thus translate
  state preparation of $\ket{+}$ to the classical scenario by preparing $0$ and
  then adding a random bit $r$ to it, which we assume to be generated in a
  leak-free manner, so that it's hidden from the adversary.
  
  The output doesn't correspond univocally to $\ket{+}$; if we had prepared
  $\ket{-}$, we'd get the same output. But we can show that from the point of
  view of a (classical) adversary that can see all the wires in the circuit,
  except for the leak-free component, this preparation procedure is equivalent
  to one that prepares $\ket{+}$, shown in Figure~\ref{fig:prepp}.

  \begin{figure}[h!]
    \centering
    \begin{displaymath}
      \Qcircuit @C=.5em @R=.3em @!R {
        & & \lstick{\ket{0}} & \targ & \push{I} \qw & \qw & \gate{Z} & \rstick{\ket{+}} \qw \\
        \push{\rule{.5em}{0em}} & & \lstick{\ket{+}} & \ctrl{0} \qwx & \gate{H} & \meter & \control \cw \cwx \gategroup{2}{1}{2}{7}{1.4em}{--}
      }
    \end{displaymath}
    \caption{Preparation of $\ket{+}$.}
    \label{fig:prepp}

  \end{figure}
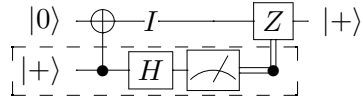

  The second wire, shown inside the dashed box, is in the leak-free part of the
  circuit. The state at the point $I$ is $\kb{+} + \kb{-}$; this is the state
  used in the actual (classical) circuit. The figure shows how we could then
  correct this state using the leak-free component, so that in the end we get
  $\ket{+}$. Since the correction operation has no effect on the $Z$-basis
  value, it can be omitted. Therefore if we use the circuit in
  Figure~\ref{fig:prepp} for preparation of $\ket{+}$ in the circuit
  $\mathcal{C}$, the classical translation amounts to initializing a register to
  0 and then adding a (leak-free) random bit to it.

  Measurement in the $X$ basis can be done in a similar way.  Measuring $X$
  projects the state $\ket{\psi}$ onto one of the operator's eigenspaces; the
  projection operators are given by $P_i = \frac12(\mathbbm 1 + (-1)^i X)$ for
  $i \in \{0,1\}$ and the post-measurement state is given by $\rho =
  \sum_{i=0}^1 P_i \kb{\psi} P_i = \frac12\sum_{i=0}^1 X^i
  \kb{\psi} X^i$. Hence measuring in the $X$ basis is equivalent to
  a random bit flip, which we can simulate in the classical circuit by adding a
  random bit.
\end{proof}

The significance of this theorem is that we can use these
operations~(Section~\ref{sec:lrg}) to implement a reliable Toffoli gate, a
reversible circuit that is universal for classical computation, for
\textit{independent (phase) noise}. This noise consists of letting every wire in
the circuit be subject to a phase error\footnote{A phase error is an event where
  a phase flip can be introduced with probability $1/2$.} with a fixed
probability $p$. We'll see that this corresponds to the leakage model of
independent leakage, where each wire in the circuit can leak with probability
$p$.

\section{Leakage-resilient gadgets}
\label{sec:lrg}

Our fault-tolerant construction follows~\cite{aliferis2005quantum}, which works
in the model of independent phase noise where each wire in the circuit is
subject to a phase error with probability $p$. This phase noise model is related
to the independent leakage model, where each wire in the circuit leaks with
probability $p$. Leaking the state of one wire is equivalent to introducing one
phase error, as seen in Figure~\ref{fig:leak}.

  \begin{figure}[h!]
    \centering
    \begin{displaymath}
      \begin{aligned}
        \Qcircuit @C=1em @R=1em @!R {
          & \ctrl{1} & \qw \\
          \lstick{\ket{0}} & \targ & \qw }
      \end{aligned}
      \hspace{1em} = \hspace{2em}
      \begin{aligned}
        \Qcircuit @C=1em @R=1em @!R {
          & \qw & \ctrl{1} & \qw & \qw \\
          \lstick{\ket{0}} & \gate{H} & \gate{Z} & \gate{H} & \qw }
      \end{aligned} \hspace{1em} = \hspace{2em}
      \begin{aligned}
        \Qcircuit @C=1em @R=1em @!R {
          & \gate{Z} & \qw & \qw \\
          \lstick{\ket{+}} & \ctrl{-1} & \gate{H} & \qw }
      \end{aligned}
    \end{displaymath}
    
    \caption{Leaking one bit (encoded into the top wire) is equivalent to introducing a phase error.}
\label{fig:leak}
\end{figure}
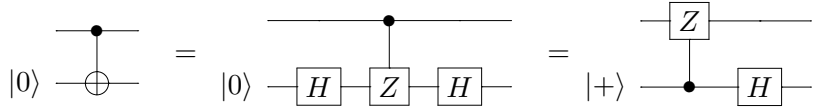

This relation also provides a more direct relation between fault tolerance and
leakage resilience than the one given in Theorem~\ref{thm:ft-lr}: No bits can
leak if all the errors are corrected, and therefore an
$\varepsilon$-reliable circuit is $\varepsilon$-leakage-resilient.

In order to construct a leakage-resilient compiler, we just need to implement a
set of fault-tolerant gates that is universal for classical computation. We only
use the components in the statement of Theorem~\ref{thm:c-q}, so that we are
able to make the implementation classical. A typical approach is to implement
the NAND gate. However, quantum gates are reversible, hence it's easier to
implement a universal set of reversible gates. We have chosen the Toffoli gate,
defined by $T(a,b,c) = (a,b,ab \oplus c)$, which is universal.

The construction we develop in the remainder of this section only uses the
components specified in Theorem~\ref{thm:c-q}. By this theorem, for each
component there exists a classical translation with the same wires as the
original quantum component, and therefore we can apply the relationship between
independent phase noise and independent leakage explained above directly (or, in
general, we could apply Lemma~\ref{lem:phasenoise}). We have the following
result.

\begin{corollary}
  \label{cor:uni}
  Let $L$ be the independent leakage model and assume the probability of leakage
  $p$ satisfies $p < 10^{-5}$. For any $\varepsilon > 0$ there exists a
  $\varepsilon$-leakage-resilient compiler $\pi$ that, given an arbitrary
  reversible classical circuit $\mathcal{C}$ with $l$ locations and depth $d$ as
  input, outputs a classical circuit $\overline{\mathcal{C}}$ with $l' = O(l
  \operatorname{polylog}(l/\varepsilon))$ and depth $d' = O(d
  \operatorname{polylog}(l/\varepsilon))$.
\end{corollary}

This corollary follows from the implementation of the Toffoli gate and the
accuracy threshold theorem of~\cite{aliferis2005quantum}, along with
Theorem~\ref{thm:c-q} and the relationship between independent phase noise and
independent leakage.

Given the construction for the Toffoli gate, we have a method of compiling an
arbitrary circuit into a private one that works as follows. First, we convert
the circuit into one using only Toffoli gates. As we'll see, our fault-tolerant
implementation of the Toffoli gate involves only $X,Z$ and CNOT gates, state
preparation, measurement in the $X$ and $Z$ bases, and error correction. Hence
it has a classical translation due to Theorem~\ref{thm:c-q}. Furthermore, by
Theorem~\ref{thm:ft-lr}, the resulting circuit is leakage-resilient against
independent leakage.

The construction in~\cite{aliferis2005quantum} works by encoding qubits in the
Steane $[[7,1,3]]$ code~\cite{steane1996multiple}. Fault tolerance is achieved
by constructing encoded gadgets that are resilient to errors ``in first order''
---that is, if the physical circuit has probability of failure $\varepsilon$,
then the encoded circuit has probability of failure $O(\varepsilon^2)$. One can
then achieve an arbitrary degree of accuracy by code concatenation.

In what follows, we present the fault-tolerant gadgets we use in the
construction of the Toffoli gate.

\vspace{.5em}

\noindent \textbf{Measurement.}  A $Z$ basis measurement can be done by
measuring transversally. For the Steane code, $\enc{0}$ is an equal
superposition of all the even-weight codewords of the Hamming code, and
$\enc{1}$ is an equal superposition of the odd-weight codewords
(Section~\ref{sec:steane-code} has a review of the Steane code). Thus we can
determine which state was prepared by computing the parity of the measurement
outcomes. Since phase errors don't affect the outcomes, the procedure is
fault-tolerant.

$X$ basis measurement can thus be done by applying $\overline{X} = X_1 X_2 X_3$
controlled by $\ket{+}$ and then measuring in the $Z$ basis. Phase errors only
propagate from the data to the ancilla $\ket{+}$ state, which, as argued in the
proof of theorem~\ref{thm:c-q}, can be assumed to be prepared perfectly, since
classically they correspond to generating a random bit. This ensures that the
procedure is fault-tolerant.

\vspace{.5em}

\noindent \textbf{Error correction.} 
Error correction consists
  of syndrome extraction, in which the errors are diagnosed, and a recovery
  step, performed in order to transform the state back into the correct
  one. First, notice that since we're performing computation on $Z$ basis states
  and we assume that only phase noise is possible, the recovery step would always 
  consist of applying $Z$ operators and by the remarks in Theorem~\ref{thm:c-q} can never
  change the outcome of the computation. Therefore only the syndrome extraction
  step is necessary. For that we use a method known as Steane error
  correction~\cite{steane1997active,gottesman2009introduction}, for which syndrome extraction reduces to
  $X$ basis measurement.

\vspace{.5em}

\noindent \textbf{CNOT and $X$ gates.}
For the Steane code, the $X$ gate can be applied transversally, that is,
$\overline{X} = \otimes_{i=1}^7 X_i$. The CNOT gate can also be easily seen to
be transversal, in the sense that every qubit in the first block only interacts
with the corresponding qubit in the second block.

\vspace{.5em}

\noindent \textbf{State preparation.}  We first describe the state preparation
of $\enc{0}$. It is accomplished by preparing 7 $\ket{0}$ states and then
performing Steane (phase) error correction. This method consists of taking the
state $\ket{0}^{\otimes 7}$ and then performing a CNOT gate with this state as
target and an ancilla state $\enc{0}$ as control. The ancilla is then measured
in the $Z$ basis. This, of course, has an obvious problem---in order to execute
the circuit, we need $\enc{0}$, which is exactly the state we are trying to
prepare. But before addressing this issue, let us see why it works. The outcomes
of the measurement of $\enc{0}$ after the transversal CNOT with
$\ket{0}^{\otimes 7}$ determine the eigenvalues of the $X$-type stabilizers of
the code. The original state $\ket 0^{\otimes 7}$ is projected onto the subspace
associated with the particular measurement outcomes. We can then map the state
into the codespace by changing the eigenvalues as needed. But in fact recovery
is not necessary, since it could only change the phase of the state. As long as
we have the information about the eigenvalues, we might as well adopt the
resulting state as $\enc{0}$.

In order to prepare the ancilla $\enc{0}$, we use the circuit shown
in~\cite[Fig.~12]{aliferis2005quantum}. We note that, since the circuit
acts on physical qubits, it's not fault-tolerant. The usual way to prepare an
ancilla state fault-tolerantly is to perform a verification step after encoding,
where the state is rejected and the procedure repeated if the verification
detects too many errors. However, this doesn't help us because there's no clear
classical analogue; the circuit wouldn't ``know'' when to reject a state, since
phase errors don't show up in the classical picture.

Instead, we use the ancilla verification method developed
in~\cite{divincenzo2007effective}. In this method, the ancilla state is
prepared, interacts with the data transversally and is subsequently decoded. The
decoding is done in such a way that errors due to a single fault in the ancilla
preparation can be perfectly distinguished and the data block can be
corrected. In our case, we can decode by measuring the phase-error syndrome,
which for the Steane code reduces to $X$ basis measurement.

Now we can use $\enc{0}$ along with $\ket{+}$ to prepare $\enc{+}$: we prepare
$\enc{0}$, and then apply $\overline{X}$ controlled by a random bit. As we've
argued for $X$ basis measurement, the circuit is fault-tolerant.

The method used to prepare $\enc{0}$ can also be used to prepare the ``Shor
state'' $\shor$, which is a superposition of all the even-weight words in
$\{0,1\}^7$ and is used in the construction of the Toffoli ancilla state
(below). Our method to prepare and verify $\shor$ is the same as the one
presented in~\cite{divincenzo2007effective} for the cat state, except it's done
in the rotated basis ($\shor$ is obtained from the cat state by applying the
Hadamard gate transversally). For more details, see
Appendix~\ref{sec:shor-state}.

\vspace{.5em}

\noindent \textbf{Toffoli gate.}
The gadget for the Toffoli gate is shown in Figure~\ref{fig:toffoli-q}. The
correctness of the circuit can be verified by inspection. However we've not
shown how to execute the subcircuit in the dashed rectangle; indeed, this
subcircuit uses a Toffoli gate, which is exactly what we're trying to
implement. Instead, we use an alternative circuit to prepare the state
$\enc{\Theta} = \enc{000} + \enc{100} + \enc{010} + \enc{111}$, which is the
output of the subcircuit in the dashed rectangle. The alternative circuit is
shown in Figure~\ref{fig:toffoli-a}.

\begin{figure}[h!]
  \begin{subfigure}[b]{0.6\textwidth}
    \begin{displaymath}
      \Qcircuit @C=.5em @R=.5em @!R {
        \push{\rule{1em}{0em}} & & \lstick{\enc{+}} &\control \qw \qwx[1] & \ctrl{3} & \qw & \qw & \qw &
        \ctrl{1} & \qw & \ctrl{2} & \gate{X} & \qw & \rstick{\overline{x}} \qw \\
        & & \lstick{\enc{+}} & \ctrl{1}          & \qw & \ctrl{3} & \qw & \qw &
        \gate{Z} & \gate{X} & \qw & \qw & \ctrl{1} & \rstick{\overline{y}} \qw \\
        & & \lstick{\enc{0}} & \targ             & \qw & \qw & \targ & \gate{Z}
        & \qw & \qw & \targ & \qw & \targ & \rstick{\overline{z+xy}} \qw \\
        & & \lstick{\enc{x}} & \qw               & \targ & \qw & \qw & \qw &
        \qw & \qw & \qw & \measure{Z} \cwx[-3] & \control \cw \cwx & \cw \\
        & & \lstick{\enc{y}} & \qw               & \qw & \targ & \qw & \qw &
        \qw & \measure{Z} \cwx[-3] & \control \cw \cwx[-2] & \cw & \cw & \cw \\
        & & \lstick{\enc{z}} & \qw               & \qw & \qw & \ctrl{-3} &
        \measure{X} \cwx[-3] & \control \cw \cwx[-4] \gategroup{1}{1}{3}{4}{1em}{--}
        & \cw & \cw & \cw & \cw & \cw
      }
    \end{displaymath}
    \caption{}
    \label{fig:toffoli-q}
  \end{subfigure}
  \begin{subfigure}[b]{0.4\textwidth}
    \begin{displaymath}
      \Qcircuit @C=.5em @R=.5em @!R {
        \lstick{\enc{+}} & \qw & \control \qw \qwx[1] & \qw & \qw & \qw \\
        \lstick{\enc{+}} & \qw & \ctrl{2} & \qw & \qw & \qw \\
        \lstick{\enc{+}} & \ctrl{1} & \qw & \gate{X} & \qw & \qw \\
        \lstick{\ket{\text{even}_7}} & \targ & \targ & \measure{\mbox{parity}} \cwx
      }
    \end{displaymath}
    \caption{}
    \label{fig:toffoli-a}
  \end{subfigure}
  \caption{(\subref{fig:toffoli-q}) Reliable Toffoli
    gate. (\subref{fig:toffoli-a}) Construction of the Toffoli ancilla
    state. All the gates are executed transversally. After computing the parity
    the circuit~\ref{fig:shor-dec} is executed on the bottom block.}
\end{figure}
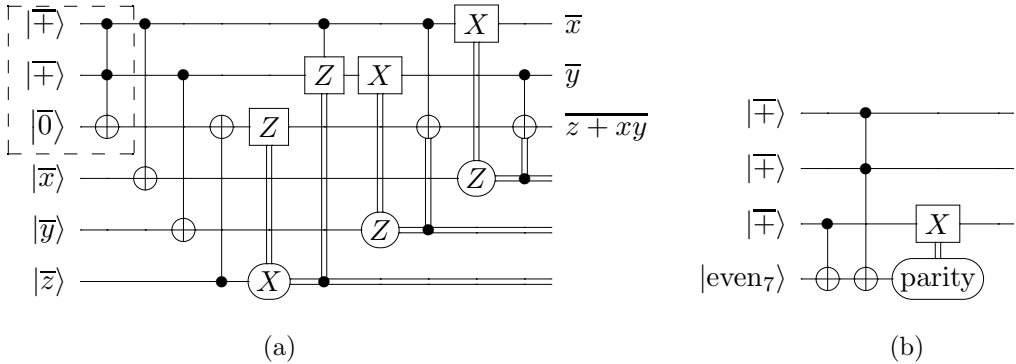

All the gates in Figure~\ref{fig:toffoli-a} are executed transversally, so the
circuit is fault-tolerant.\footnote{Because the Toffoli gate is not a Clifford
  gate, it doesn't propagate Pauli errors to Pauli errors; in particular, a
  phase flip in the third block is propagated to a superposition of phase flips.
  This would make a classical simulation of the level-$k$ circuit inefficient,
  because the error patterns to keep track of grow exponentially on $k$, but
  again, this is not a problem for our construction because we don't need to
  perform error recovery.}  We'll show that if transversal CNOT and Toffoli work
as the ``encoded'' circuits in this case, where $\ket{\text{even}_7}$ encodes 0
and the ``odd'' state $\ket{\text{odd}_7}$ encodes 1, then the circuit prepares
$\enc{\Theta}$. Specifically, we need to check that output in the target of the
CNOT is $\ket{\text{even}_7}$ if the control is $\enc{0}$ and
$\ket{\text{odd}_7}$ if the control is $\enc{1}$, and that the Toffoli gate
outputs $\ket{\text{odd}_7}$ if and only if both control qubits are $\enc{1}$,
and $\ket{\text{even}_7}$ otherwise. For the CNOT, $\enc{0}$ is an equal
superposition of even-weight strings, so XORing them with even-weight strings
outputs even-weight strings. $\enc{1}$ is an equal superposition of odd-weight
strings, and XORing them with even-weight strings produces odd-weight strings.

For the Toffoli gate, it suffices to prove that the transversal product of
$\enc{x}$ and $\enc{y}$ is a superposition of odd-weight strings if and only if
$x = y = 1$, and is a superposition of even-weight strings otherwise. The key to
show this claim is that the Steane code is based on a classical code
$\mathcal{C}$ that contains its own dual $\mathcal{C}^{\perp}$. Let $H$ be the
parity-check matrix of the code $C$. We have $H \mathbf{c} = 0$ for every
$\mathbf{c} \in \mathcal{C}$. That is, each codeword is orthogonal to every row
$H_i$ of $H$. Since $\mathcal{C}^{\perp} \subset \mathcal{C}$, this also applies
to the codewords of $\mathcal{C}^{\perp}$, and in particular to $H_i$. That is,
each $H_i$ is orthogonal to itself, which can only happen if it has even
weight. Hence all the codewords in $\mathcal{C}^{\perp}$ have even
weight. Furthermore, they must have even overlap (i.e., have 1 in the same
position) with every $H_i$ (otherwise they wouldn't be orthogonal), and
therefore also among themselves.

Thus $\enc{0}$ is the superposition of all the even-weight codewords of
$\mathcal{C}$, which have even overlap among themselves by the argument
above. The state $\enc{1}$ is the superposition of the codewords in $\mathcal{C}
- \mathcal{C}^{\perp}$, which all have odd weight. Since they also satisfy the
parity check $H$ they have even overlap with the rows of $H$, and therefore also
with the codewords in $\enc{0}$. But they have odd overlap among themselves. To
prove this, let $x,y \in \mathcal{C} - \mathcal{C}^{\perp}$. If $x$ and $y$ have
even overlap, then $x+y \in \mathcal{C}^{\perp}$ (since it has even weight) and
thus $x$ and $x+y$ would have even overlap. This proves our claim.

Thus, we have a way of preparing the ancilla state $\enc{\Theta}$
fault-tolerantly. Furthermore, the classical translation is easy: the only new
component here is measurement of the parity of the Shor state, which can be done
by measuring in the $Z$ basis and then adding up the outcomes. Now every
component in Figure~\ref{fig:toffoli-q} has a classical translation.

\section{Discussion}
\label{sec:thr}

Theorem~\ref{thm:ft-lr} establishes an explicit relationship between
fault-tolerant quantum computation and classical leakage resilience. Although it
relates leakage and noise models (via Lemma~\ref{lem:phasenoise}), it is not
clear how the properties of the noise model relate to the properties of the
leakage model in general. However, in some cases the leakage model resulting
from a given noise model has a simple interpretation. In this work, we analyzed
the independent leakage model, which corresponds to the independent phase noise
model. A possible further direction is to take a leakage model that is used in
other leakage resilience proposals and try to understand the corresponding noise
model. Conversely, one could take, say, the ``local noise'' model
of~\cite{aliferis2005quantum} and see what the corresponding leakage model looks
like.

We developed a concrete implementation of universal leakage-resilient
computation based on the fault-tolerant construction
of~\cite{aliferis2005quantum}. This construction works in the independent noise
model by using the concatenated Steane code. Fault tolerance is achieved
assuming probability of error per wire or gate $p < 10^{-5}$ (the accuracy
threshold). Along with our results, this gives us a leakage-resilient
construction for universal classical computation assuming independent leakage
with probability of leakage per wire $p < 10^{-5}$.

Taking into account the fact that we only want leakage resilience rather than
fault tolerance, this threshold can be improved. An accuracy threshold for
concatenated codes depends essentially on the size of the largest gadget in the
first level of encoding (a level-1 extended rectangle or 1-exRec in the language
of~\cite{aliferis2005quantum}). A lower bound for the threshold is the
reciprocal of the number of pairs of locations in the largest 1-exRec. In the
implementation of~\cite{aliferis2005quantum}, the largest 1-exRec is the CNOT
gate, but due to the simplifications in our case, the largest extended
rectangles are the gadgets for state preparation: they have 20 locations
each. The number of pairs of locations is then $\binom{20}{2} = 190$, which
gives a crude threshold estimate of $p \approx 0.5\%$.

As in~\cite{aliferis2005quantum}, we only use the Steane code, but we note that
our construction works for any CSS code based on a dual-containing classical
code. A promising possibility is to use color
codes~\cite{bombin2006topological}, for which numerical evidence suggests they
have good accuracy thresholds~\cite{landahl2011fault} but currently lack a
rigorous lower bound on the threshold.

We note that while the independent leakage model looks similar to the ``noisy
leakage'' case of~\cite{faust2010protecting}, they're in fact different
models. In the noisy leakage model, all bits leak, but the adversary only
receives noisy versions of these bits, each of them being flipped with
probability $p$. Crucially, the adversary doesn't know which bits have leaked
faithfully, whereas in our model every bit that's leaked is sure to arrive
correctly at the adversary.

While we don't expect existing results on fault tolerance to give direct
constructions for the leakage models commonly studied in the literature
(especially since fault tolerance has only been shown to work in very few noise
models), we note that quantum fault tolerance is a stronger requirement than
classical leakage resilience: as we've seen in this work, translating from the
former to the latter allows us to make several simplifications. Additionally,
the techniques used in the area of fault tolerance are different from those used
for leakage resilience.  Hence, we expect our result to shed new light on
leakage resilience.

\vspace{0.5em}

\noindent \textbf{Acknowledgments.} The authors would like to thank Christopher
Portmann for helpful comments on the abstract cryptography framework. This work
was supported by the Swiss National Science Foundation (through the National
Centre of Competence in Research `Quantum Science and Technology' and grant
No. 200020-135048), by the European Research Council (grant No. 258932) and by
CNPq/Brazil.

\bibliography{ft}{}
\bibliographystyle{amsalpha}

\appendix

\section{Qubits and Stabilizers}
\label{sec:app}

A quantum system can be described by a complex Hilbert space. In this work, we
deal only with two-level systems. Let $\mathcal{H}_2$ be a two-dimensional
complex Hilbert space and let $\{\ket{0},\ket{1}\}$ denote one orthonormal
basis, which we call the standard basis or the computational basis. An
arbitrary element $\ket{\psi}$ of $\mathcal{H}_2$ can be written $\ket{\psi} =
\alpha \ket{0} + \beta \ket{1}$ were $\alpha,\beta \in \mathbf{C}$. Quantum
states are represented by unit vectors, that is, $|\alpha|^2 +
|\beta|^2 = 1$.

Composite systems are given by the tensor product, that is, if $A,B$ are two
complex Hilbert spaces, then the state of the system composed by $A$ and $B$ is
an element of $A \otimes B$.

Systems whose state is not completely known are described properly by
\textit{density operators}. Suppose a system is in the state $\ket{\psi_i}$ with
probability $p_i$. We define the density operator $\rho$ for the system by
\begin{displaymath}
  \rho = \sum_i p_i \kb{\psi_i}
\end{displaymath}

States that are completely known, that is, states of the form $\rho =
\kb{\psi}$, are called \textit{pure states}, while other states are called
\textit{mixed states}.

In the following, as well as in the rest of the text, we only deal with
states that are equal superpositions, so we omit the normalizing factor. So,
for instance, we write the state $\ket{-} = \frac{1}{\sqrt{2}} (\ket{0} -
\ket{1})$ simply as $\ket{0} - \ket{1}$.

\subsection{Pauli operators}
\label{sec:pauli-operators}

Take the standard basis for $\mathcal{H}_2$ and let
\begin{displaymath}
  \mathbbm{1} =
  \begin{pmatrix}
    1 & 0 \\
    0 & 1
  \end{pmatrix}
  \qquad
  X =
  \begin{pmatrix}
    0 & 1 \\
    1 & 0
  \end{pmatrix}
  \qquad
  Y =
  \begin{pmatrix}
    0 & -i \\
    i & 0
  \end{pmatrix}
  \qquad
  Z =
  \begin{pmatrix}
    1 & 0 \\
    0 & -1
  \end{pmatrix}
\end{displaymath}

$X,Y,Z$ are called the \textit{Pauli matrices}. They anticommute with each
other. Furthermore, an arbitrary matrix $A$ on $\mathcal{H}_2$ can be written $A
= a_0 \mathbbm{1} + a_1 X + a_2 Y + a_3 Z$ with $a_0,a_1,a_2,a_3 \in
\mathbf{C}$. Similarly, operators on $\mathcal{H}_2^{\otimes n}$ can also be
written in terms of Pauli matrices. Let $\mathcal{P}_n$ be the set of operators
of the form $i^k \bigotimes_{i=1}^n P_i$ where $k \in \mathbf{Z}$ and $P_i \in
\{\mathbbm{1},X,Y,Z\}$. The set $\mathcal{P}_n$ is a nonabelian group; it's
called the \textit{Pauli group} on $n$ qubits. Because $Y = i X Z$, the group is
generated by $X$ and $Z$ (up to a phase factor); that is, $\mathcal{P}_n =
\langle i\mathbbm{1}, X_1, \dots, X_n, Z_1, \dots, Z_n \rangle$, where $X_i$
denotes the operator that acts as $X$ on the $i$th qubit, and similarly for
$Z_i$. An arbitrary matrix on $\mathcal{H}_2^{\otimes n}$ can be written as a
linear combination of elements of $\mathcal{P}_n$.

Consider now a system $S$ subject to noise from the environment $E$. By the
above discussion, the system evolves as
\begin{displaymath}
  \ket{\psi}_S \ket{0}_E \to \sum_k E_k \ket{\psi}_A \ket{e_k}_E
\end{displaymath}
where the states $\ket{e_k}$ are not necessarily orthogonal, and $\{E_k\}$ is a
set of linearly independent Pauli operators. We call each of them a
\textit{Pauli error}. We say a Pauli error $E_k$ has weight $t$ if it acts
nontrivially on at most $t$ qubits. Because general errors can always be
decomposed into Pauli errors, we only need to design error correcting codes that
can correct Pauli errors.

\subsection{Error-correcting codes}
\label{sec:error-corr-codes}

A quantum error-correcting code $\mathcal{C}$ of length $n$ is a subspace of
$\mathcal{H}_2^{\otimes n}$. Let $\mathcal{E}$ be a set of errors. We say that
$\mathcal{C}$ corrects $\mathcal{E}$ if there exists a recovery operator $R$
acting on a larger system $\mathcal{C} \otimes A$ such that for every
$\ket{\psi} \in \mathcal{C}$ and every $E \in \mathcal{E}$, we have $\tr_A R(E
\ket{\psi}_{\mathcal{C}} \ket{a}_A) = \ket{\psi}_{\mathcal{C}}$, where $\ket{a}$
is some ancilla state in $A$.

A code that corrects all errors of weight $t$ can \textit{detect} all errors of
weight $2t$. We define the \textit{distance} of a code as the weight of the
error of smallest weight that isn't detectable. Thus, a code that can correct
$t$ errors has distance $d = 2 t + 1$. Analogously to the classical case, we
call an error correcting code of weight $n$ that encodes $k$ qubits and has
distance $d$ an $[[n,k,d]]$ code.

Now let $\mathcal{S}$ be a subgroup of the Pauli group $\mathcal{P}_n$ not
containing $-\mathbbm{1}$. Given such a subgroup, we define a \textit{stabilizer
  code} $\mathcal{C}$ of length $n$ as
\begin{displaymath}
  \mathcal{C} = \{\ket{\psi} \in \mathcal{H}_2^{\otimes n}\,\vert\,s \ket{\psi} =
  \ket{\psi}, \forall s \in \mathcal{S}\}.
\end{displaymath}

The group $\mathcal{S}$ is called the code's stabilizer. If the code encodes $k$
qubits, the stabilizer is generated by $n-k$ elements. Errors that are not
detectable are in the centralizer $\mathcal{Z}(\mathcal{S})$, the group of
operators that commute with all elements of $\mathcal{S}$. Since errors act
nontrivially on the codewords, they're not in $\mathcal{S}$; thus, the distance
of the code is given by the weight of the operator in $\mathcal{Z}(\mathcal{S})
- \mathcal{S}$ with smallest weight.

Since $\mathcal{P}_n$ has dimension $2n$, the centralizer has dimension $2n -
(n-k) = n + k$. Therefore $\mathcal{Z}(\mathcal{S}) - \mathcal{S}$ has dimension
$2k$. These operators can be regarded as \textit{logical operations} on the
codewords. We can always choose them to be the operators
$\overline{Z_1},\dots,\overline{Z_k},\overline{X_1},\dots,\overline{X_k}$,
satisfying the anticommutation relation $\overline{Z_i}\, \overline{X_j} =
(-1)^{\delta_{ij}} \overline{X_j}\,\overline{Z_i}$.

Stabilizer generators are the quantum analogue of rows in the parity-check
matrix of a classical code. In fact, there's a general class of codes known as
\textit{CSS codes} that are constructed from two classical codes $C_1,C_2$ such
that $C_2^\perp \subset C_1$. We won't introduce the general theory of CSS codes
here, and instead concern ourselves with a particular CSS code known as the
Steane code.

\subsection{The Steane code}
\label{sec:steane-code}

The Steane code~\cite{steane1996multiple} is a $[[7,1,3]]$ code that has
stabilizer generators
\begin{align*}
  g_1 &= Z_1 Z_3 Z_5 Z_7 \\
  g_2 &= Z_2 Z_3 Z_6 Z_7 \\
  g_3 &= Z_4 Z_5 Z_6 Z_7 \\
  g_4 &= X_1 X_3 X_5 X_7 \\
  g_5 &= X_2 X_3 X_6 X_7 \\
  g_6 &= X_4 X_5 X_6 X_7
\end{align*}
and logical operators
\begin{align*}
  \overline{Z} &= Z_1 Z_2 Z_3 \\
  \overline{X} &= X_1 X_2 X_3
\end{align*}

The Steane code is a CSS code based on the classical $[7,4,3]$ Hamming code,
which we denote here by $C$. The Hamming code contains its own dual, that is,
$C^\perp \subset C$. The stabilizer generators $g_1,g_2,g_3$ and $g_4,g_5,g_6$
for the Steane code correspond to the rows in the parity check matrix for $C$,
where a 0 in the parity check matrix corresponds to the identity and 1
corresponds to $Z$ (for the $Z$-type stabilizers) or $X$ (for the $X$-type
stabilizers). The codewords are given by
\begin{align*}
\enc{0} = \textstyle{\sum_{x \in C^\perp}} \ket{x} &= \ket{0000000} + \ket{0001111} +
\ket{0110011} + \ket{1010101} \\
 & \qquad {} + \ket{0111100} + \ket{1011010} + \ket{1100110} \\
\enc{1} = \textstyle{\sum_{x \in C - C^\perp}} \ket{x} &= \ket{1111111} + \ket{1110000} +
\ket{1001100} + \ket{0101010} \\
& \qquad {} + \ket{1000011} + \ket{0100101} + \ket{0011001}.
\end{align*}

Notice that $\enc{0}$ is an equal superposition of all the even-weight Hamming
codewords, and $\enc{1}$ is an equal superposition of all the odd-weight
ones. This comes from the fact that $C^\perp \subset C$, which implies several
properties that make the Steane code useful for fault tolerance. In particular,
the logical Hadamard gate can be implemented transversally, that is,
$\overline{H} = H^{\otimes 7}$. This fact is also crucial in the construction we
use for the Toffoli ancilla state, used to implement the Toffoli gate.

\subsection{Quantum gates and measurement}
\label{sec:quantum-gates}

$X$ and $Z$ gates correspond to the application of the respective Pauli
matrices. CZ and CNOT are the controlled version of these gates, that is,
they're linear operators with
\begin{align*}
  & \text{CNOT}(\ket{a},\ket{b}) = (\ket{a},\ket{a+b}) \\
  & \text{CZ}(\ket{a},\alpha \ket{0} + \beta \ket{1}) = (\ket{a},\alpha \ket{0}
  + \beta (-1)^a \ket{1})
\end{align*}

Measurements are operations that take a quantum state as input and have a
classical outcome. A measurement in the $Z$ basis of an arbitrary qubit $\alpha
\ket{0} + \beta \ket{1}$ returns $0$ with probability $|\alpha^2|$ and $1$
with probability $|\beta^2|$. A measurement in the $X$ basis works similarly,
by writing the state in the $X$ basis: measuring $\alpha \ket{+} + \beta
\ket{-}$ gives $0$ with probability $|\alpha^2|$ and $1$ with probability
$|\beta^2|$.

\section{Construction of the Shor state}
\label{sec:shor-state}

The circuit to construct the Shor state is shown in Figure~\ref{fig:shor}. After
it interacts with the data, we measure its syndrome using the circuit in
Figure~\ref{fig:shor-dec} in order to diagnose a possible fault in the
preparation.

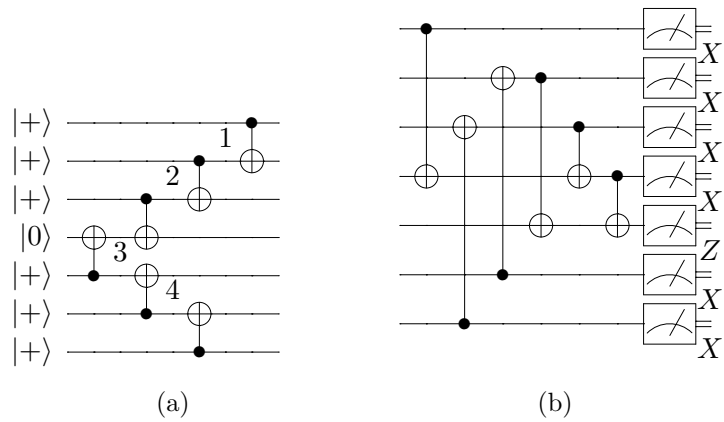
\begin{figure}[h!]
  \centering
  \begin{subfigure}[b]{0.3\textwidth}
    \begin{displaymath}
      \Qcircuit @C=.5em @R=.5em @!R {
        \lstick{\ket{+}} & \qw & \qw & \qw & \qw & \qw & \qw & \ctrl{1} & \qw \\
        \lstick{\ket{+}} & \qw & \qw & \qw & \qw & \ctrl{1} & \ustick{1} \qw & \targ & \qw \\
        \lstick{\ket{+}} & \qw & \qw & \ctrl{1} & \ustick{2} \qw & \targ & \qw & \qw & \qw \\
        \lstick{\ket{0}} & \targ & \qw & \targ & \qw & \qw & \qw & \qw & \qw \\
        \lstick{\ket{+}} & \ctrl{-1} & \ustick{3} \qw & \targ & \qw & \qw & \qw & \qw & \qw \\
        \lstick{\ket{+}} & \qw & \qw & \ctrl{-1} & \ustick{4} \qw & \targ & \qw & \qw & \qw \\
        \lstick{\ket{+}} & \qw & \qw & \qw & \qw & \ctrl{-1} & \qw & \qw & \qw
      }
    \end{displaymath}
    \caption{}
    \label{fig:shor}
  \end{subfigure}
  \begin{subfigure}[b]{0.3\textwidth}
    \begin{displaymath}
      \Qcircuit @C=.5em @R=.3em @!R {
        & \ctrl{3} & \qw & \qw & \qw & \qw & \qw & \meter & \dstick{X} \cw \\
        & \qw & \qw & \targ & \ctrl{3} & \qw & \qw & \meter & \dstick{X} \cw \\
        & \qw & \targ & \qw & \qw & \ctrl{1} & \qw & \meter & \dstick{X} \cw \\
        & \targ & \qw & \qw & \qw & \targ & \ctrl{1} & \meter & \dstick{X} \cw \\
        & \qw  & \qw & \qw & \targ & \qw & \targ & \meter & \dstick{Z} \cw \\
        & \qw & \qw & \ctrl{-4} & \qw & \qw & \qw & \meter & \dstick{X} \cw \\
        & \qw & \ctrl{-4} & \qw & \qw & \qw & \qw & \meter & \dstick{X} \cw
      }
    \end{displaymath}
    \caption{}
    \label{fig:shor-dec}
  \end{subfigure}
  \caption{(\subref{fig:shor}) Preparation of the Shor state $\shor$. The
    numbers indicate positions where errors introduced there cause multiple
    errors in the output. (\subref{fig:shor-dec}) Syndrome measurement for the
    Shor state.}
\end{figure}

We need to show that there's a decoding procedure such that any possible errors
introduced in the encoding of the Shor state leads to different error
syndromes, which could then be used to correct any errors it might have caused
in the circuit. First let's look at all the possible multiple $Z$ error patterns
resulting from a single fault that might occur during encoding. Those are shown
by the numbers in Figure~\ref{fig:shor}. All the other possibilities either lead
to a single error in the output or to the same error pattern as one of those.

The error patterns, in the order given by the figure, are $Z_1 Z_2$, $Z_1 Z_2
Z_3$, $Z_5 Z_6 Z_7$ and $Z_6 Z_7$. The decoding circuit
(Figure~\ref{fig:shor-dec}) propagates these errors to $Z_1 Z_2 Z_6$, $Z_1 Z_2
Z_3 Z_6 Z_7$, $Z_2 Z_4 Z_6 Z_7$ and $Z_6 Z_7$ respectively. Thus all possible
error patterns give different syndromes, making them correctable.

\end{document}

%% file: lr-ideal.tex
\begin{tikzpicture}[node distance=2.5cm,auto,shorten >=1pt,scale=0.7, transform shape]
  \node [party] (a) {};
  \node [int, right of=a] (c) {$\mathcal{S}$};
  \node [party, right of=c] (e) {};

  \coordinate (aue) at ($(a.mid east) + (0,1)$);
  \coordinate (ale) at ($(a.mid east) - (0,1)$);
  \path[->] (aue) edge node {$y$} (aue-|c.west);
  \path[->] (ale) edge node {$\mathcal{C}$} (ale-|c.west);

  \coordinate (euw) at ($(e.mid west) + (0,1.2)$);
  \coordinate (cue) at ($(c.mid east) + (0,1)$);
  \coordinate (cle) at ($(c.mid east) + (0,-1)$);
  \path[->,above] (euw) edge node {$x$} (euw-|c.east);
  \path[->,below] (cue) edge node {$\mathcal{C}(x,y)$} (cue-|e.west);
  \path[->,above] (cle) edge node {$\mathcal{C}$} (cle-|e.west);
\end{tikzpicture}

%% file: lr-real.tex
\begin{tikzpicture}[node distance=2.5cm,auto,shorten >=1pt, scale=0.7, transform
  shape]
  \node [party] (a) {};
  \node [int, right of=a] (c) {$\mathcal{R}$};
  \node [party, right of=c] (e) {};
  \coordinate (aue) at ($(a.mid east) + (0,1)$);
  \coordinate (ale) at ($(a.mid east) + (0,-1)$);
  \path[->] (aue) edge node {$y$} (aue-|c.west);
  \path[->] (ale) edge node {$\mathcal{C}$} (ale-|c.west);

  \coordinate (euw) at ($(e.mid west) + (0,1.5)$);
  \coordinate (cue) at ($(c.mid east) + (0,1.3)$);
  \coordinate (emuw) at ($(e.mid west) + (0,-.4)$);
  \coordinate (cmlw) at ($(c.mid east) + (0,-.6)$);
  \coordinate (cle) at ($(c.mid east) + (0,-1.5)$);
  \path[->,above] (euw) edge node {$x$} (euw-|c.east);
  \path[->,below] (cue) edge node {$\mathcal{C}(x,y)$} (cue-|e.west);
  \path[->,above] (emuw) edge node {$l$} (emuw-|c.east);
  \path[->,below] (cmlw) edge node {$l'$} (cmlw-|e.west);
  \path[->,below] (cle) edge node {$\mathcal{C}$} (cle-|e.west);
\end{tikzpicture}

%% file: lr-real-protocol.tex
\begin{tikzpicture}[node distance=2.8cm,auto,shorten
  >=1pt,baseline={([yshift=-.5ex]current bounding
    box.center)},scale=0.8,transform shape]
  \node [party] (a) {};
  \node [int, right of=a] (p) {$\pi_A$};
  \node [int, right of=p] (c) {$\mathcal{R}$};
  \node [party, right of=c] (e) {};
  \coordinate (aue) at ($(a.mid east) + (0,1)$);
  \coordinate (ale) at ($(a.mid east) + (0,-1)$);
  \path[->] (aue) edge node {$y$} (aue-|p.west);
  \path[->] (ale) edge node {$\mathcal{C}$} (ale-|p.west);

  \coordinate (pue) at ($(p.mid east) + (0,1)$);
  \coordinate (ple) at ($(p.mid east) + (0,-1)$);
  \path[->] (pue) edge node {$\overline{y}$} (pue-|c.west);
  \path[->] (ple) edge node {$\overline{\mathcal{C}}$} (ple-|c.west);

  \coordinate (euw) at ($(e.mid west) + (0,1.5)$);
  \coordinate (cue) at ($(c.mid east) + (0,1.3)$);
  \coordinate (elw) at ($(e.mid west) + (0,-.4)$);
  \coordinate (cle) at ($(c.mid east) + (0,-.6)$);
  \coordinate (clle) at ($(c.mid east) + (0,-1.5)$);
  \path[->,above] (euw) edge node {$x$} (euw-|c.east);
  \path[->,below] (cue) edge node {$\overline{\mathcal{C}(x,y)}$} (cue-|e.west);
  \path[->,above] (elw) edge node {$l$} (elw-|c.east);
  \path[->,below] (cle) edge node {$l'$} (cle-|e.west);
  \path[->,below] (clle) edge node {$\overline{\mathcal{C}}$} (clle-|e.west);
\end{tikzpicture}

%% file: lr-sim.tex
\begin{tikzpicture}[node distance=2.8cm,auto,shorten
  >=1pt,baseline={([yshift=-.5ex]current bounding
    box.center)},scale=0.8,transform shape]
  \node [party] (a) {};
  \node [int, right of=a] (c) {$\mathcal{S}$};
  \node [int, right of=c] (s) {$\sigma_E$};
  \node [party, right of=s] (e) {};

  \coordinate (aue) at ($(a.mid east) + (0,1)$);
  \coordinate (ale) at ($(a.mid east) - (0,1)$);
  \path[->] (aue) edge node {$y$} (aue-|c.west);
  \path[->] (ale) edge node {$\mathcal{C}$} (ale-|c.west);

  \coordinate (euw) at ($(s.mid west) + (0,1.2)$);
  \coordinate (cue) at ($(c.mid east) + (0,1)$);
  \coordinate (cle) at ($(c.mid east) + (0,-1)$);
  \path[->,above] (euw) edge node {$x$} (euw-|c.east);
  \path[->,below] (cue) edge node {$\mathcal{C}(x,y)$} (cue-|s.west);
  \path[->,above] (cle) edge node {$\mathcal{C}$} (cle-|s.west);

  \coordinate (euw) at ($(e.mid west) + (0,1.5)$);
  \coordinate (sue) at ($(s.mid east) + (0,1.3)$);
  \coordinate (elw) at ($(e.mid west) + (0,-.4)$);
  \coordinate (sle) at ($(s.mid east) + (0,-.6)$);
  \coordinate (slle) at ($(s.mid east) + (0,-1.5)$);
  \path[->,above] (euw) edge node {$x$} (euw-|s.east);
  \path[->,below] (sue) edge node {$\overline{\mathcal{C}(x,y)}$} (sue-|e.west);
  \path[->,above] (elw) edge node {$l$} (elw-|s.east);
  \path[->,below] (sle) edge node {$l'$} (sle-|e.west);
  \path[->,below] (slle) edge node {$\overline{\mathcal{C}}$} (slle-|e.west);

\end{tikzpicture}